\documentclass [a4paper, 11pt] {article}
\usepackage{amsmath}
\usepackage{amsthm}
\usepackage{amsfonts}
\usepackage {natbib}
\usepackage{graphicx}

\newtheorem{thm}{Theorem}

\title{Duality between Approximate Bayesian Methods and Prior Robustness}
\author{Chaitanya Joshi, Fabrizio Ruggeri}
\begin{document}
\date{April 1, 2020}
\maketitle

\begin{abstract}
In this paper we show that there is a link between approximate Bayesian methods and prior robustness. We show that what is typically recognized as an approximation to the likelihood, either due to the simulated data as in the Approximate Bayesian Computation (ABC) methods or due to the functional approximation to the likelihood, can instead also be viewed upon as an implicit exercise in prior robustness. We first define two new classes of priors for the cases where the sufficient statistics is available, establish their mathematical properties and show, for a simple illustrative example, that these classes of priors can also be used to obtain the posterior distribution that would be obtained by implementing ABC.  We then generalize and define two further classes of priors that are applicable in very general scenarios; one where the sufficient statistics is not available and another where the likelihood is approximated using a functional approximation. We then discuss the interpretation and elicitation aspects of the classes proposed here as well as their potential applications and possible computational benefits. These classes establish the duality between approximate Bayesian inference and prior robustness for a wide category of Bayesian inference methods.\\


\end{abstract}

\section{Introduction}

Bayesian analysis on complex models often involves approximations to the likelihood, either because it is necessary or else because it is convenient or computationally efficient. Some of the more widely known approximation methods include the Approximate Bayesian Computation (ABC) methods (\cite{pritchard99}, \cite{marin12}), but also the Integrated Nested Laplace Approximation (INLA) (\cite{rue09}), Variational Bayes methods (\cite{jordan99}), psuedo-likelihood (\cite{besag75}), synthetic likelihood (\cite{wood10}) and Expectation Propagation (\cite{minka01}). A significant research effort has been targeted at developing modifications to improve the approximations (see, for example, \cite{marjoram03}, \cite{marin12},   \cite{fernhead12}, \cite{martin13}) as well as studying how such an approximation may affect the accuracy of the posterior distribution thus obtained (\cite{geweke04}, \cite{cook06}, \cite{fernhead12}, \cite{prangle14}, \cite{yao18} and \cite{lee19}).

In this paper, we propose a different way of looking at such approximations. We show that there is a duality between distortion to the likelihood and distortion to the prior. Being prepared to admit a distorted (approximated) likelihood is equivalent to being  prepared to accept a distorted posterior. We argue then that the same distorted posterior can, in fact, be obtained by using the \emph{true} likelihood and a distorted prior distribution instead. That is, for every distorted likelihood, there exists a distorted prior distribution in the sense that 
$$ \mbox{\emph{distorted likelihood}} \times \mbox{true prior} \propto \mbox{true likelihood} \times \mbox{\emph{distorted prior}}.$$  Therefore, an exercise in approximate Bayesian methods can instead be viewed as an exercise in Bayesian prior robustness.

We start by focusing on the ABC methods and show how implementing an ABC approach is equivalent to instead performing a prior robustness analysis, where one uses prior distributions from the new classes of priors that we will define. We start by assuming that the sufficient statistic is available and then generalize to a more practical situation where it isn't. We then sketch a further generalization to show how a class of prior distributions can be established for other approximation methods that instead use functional approximations to the likelihood.

Let $\pi(\mathbf{\theta}),\, \mathbf{\theta} \in \Theta \subseteq \mathbb{R}^{n},$ be an $n$-dimensional prior distribution, $f(\mathbf{x}|\mathbf{\theta})$ the likelihood having observed a given data $\mathbf{x}$ and $\pi(\mathbf{\theta}|\mathbf{x})$ be the corresponding posterior distribution. So that, we have,
\begin{equation*}
\pi(\mathbf{\theta}|\mathbf{x}) \propto \pi(\mathbf{\theta}) f(\mathbf{x}|\mathbf{\theta}) \propto \pi(\mathbf{\theta}) g(s(\mathbf{x})|\mathbf{\theta}),
\end{equation*}
where $S(\cdot)$ is a sufficient statistics for $\mathbf{\theta}$ and using the factorization theorem.

For observed data $\mathbf{x^{o}},$ a typical ABC implementation accepts a particular sample value of $\theta'$ if the sufficient statistics of the  data $\mathbf{x'}$ simulated using $\theta'\sim \pi(\mathbf{\theta})$ is within an $\epsilon$-neighborhood of $s(\mathbf{x^{o}}), \; \epsilon > 0.$ That is, $s(\mathbf{x'}) \in (s(\mathbf{x^{o}}) - \epsilon, s(\mathbf{x^{o}})+\epsilon).$ In this paper, we show that this process can be viewed as an implicit exercise in prior robustness. This is because the posterior distribution obtained using $\mathbf{x'}$ can also be obtained by using the original data $\mathbf{x^{0}}$ and a new prior $\pi'$. We derive two new classes of priors: one for the general case applicable to any likelihood and prior and one specifically when the likelihood belongs to the exponential family and the prior is from a conjugate distribution. We call these classes the \emph{ABC class of priors} and the  \emph{ABC-E class of priors} respectively.

The rest of the paper is organized as follows. In Section \ref{abc-G}, we establish the ABC class of priors, discuss its mathematical properties and illustrate using two simple examples. In Section \ref{abc-E} we establish the ABC-E class of priors and its link with the ABC class of priors and also illustrate using examples. Section \ref{algo} provides importance sampling algorithms to sample from the prior distributions that belong to the ABC class of priors as well as another one to sample from the posterior distribution obtained when using the prior distribution from the ABC class. Here, we also show how the posterior obtained using an ABC method can also be obtained  by instead sampling from the posterior distributions corresponding to the prior distributions from the ABC class. We then generalize to define two further classes of priors in Section \ref{general}: the ABC-G class when the sufficient statistics is not available and the AB class when other approximate Bayesian methods that employ functional approximations to the likelihood are used. Finally, in Section 6, we conclude with a discussion on the elicitation and interpretation of the new classes of priors as well as their possible applications and possible computational advantages.


\section{ABC class of priors: General case} \label{abc-G}

For simulated data $\mathbf{x'},$ we have that
\begin{equation*} \label{Eq2-1}
\pi(\mathbf{\theta}|\mathbf{x'}) \propto \pi(\mathbf{\theta}) g(s(\mathbf{x'})|\mathbf{\theta}).
\end{equation*}
Theorem \ref{thm1} shows that the same posterior can be obtained by using the observed data $\mathbf{x^{0}}$ and a different prior distribution $\dot{\pi}$.
\begin{thm}\label{thm1}
For one-dimensional $\mathbf{\theta}$ and $s(\mathbf{x})$ and for simulated data  $\mathbf{x'},\; \pi(\mathbf{\theta}|\mathbf{x'}) \propto \dot{\pi}(\mathbf{\theta}) g(s(\mathbf{x^{0}})|\mathbf{\theta}).$ 
\end{thm}
\begin{proof}
We have that, 
\begin{equation} \label{Eq2-2}
\pi(\mathbf{\theta}|\mathbf{x'}) \propto \pi(\mathbf{\theta}) g(s(\mathbf{x'})|\mathbf{\theta}) \propto \pi(\mathbf{\theta}) \exp[\log[g(s(\mathbf{x'})|\mathbf{\theta})]].
\end{equation}
Approximating $\log[g(s(\mathbf{x'})|\mathbf{\theta})]$ using Taylor's approximation around $\log[g(s(\mathbf{x^{0}})|\mathbf{\theta})]$ we get
\begin{eqnarray} \label{Eq2-3}
\nonumber \log[g(s(\mathbf{x'})|\mathbf{\theta})] & = & \log[g(s(\mathbf{x^{0}})|\mathbf{\theta})] + \frac{d \log[g(s(\mathbf{x})|\mathbf{\theta})]}{d s(\mathbf{x})}|_{\mathbf{x}=\mathbf{x^{0}}} [s(\mathbf{x'}) - s(\mathbf{x^{0}})] + e\\
& =& \log[g(s(\mathbf{x^{0}})|\mathbf{\theta})] + \frac{g'(s(\mathbf{x^{0}})|\mathbf{\theta})}{g(s(\mathbf{x^{0}})|\mathbf{\theta})}[s(\mathbf{x'}) - s(\mathbf{x^{0}})] + e,
\end{eqnarray}

where, $ e =  o [s(\mathbf{x'}) - s(\mathbf{x^{0}})].$ Using \eqref{Eq2-2} and \eqref{Eq2-3},
\begin{eqnarray} \label{eq2-4}
\nonumber \pi(\mathbf{\theta}|\mathbf{x'}) & \propto & \pi(\mathbf{\theta}) g(s(\mathbf{x^{0}})|\mathbf{\theta})  \exp \left [\frac{g'(s(\mathbf{x^{0}})|\mathbf{\theta})}{g(s(\mathbf{x^{0}})|\mathbf{\theta})} [s(\mathbf{x'}) - s(\mathbf{x^{0}})]+ e \right ]\\
& \propto & \dot{\pi}(\mathbf{\theta}) g(s(\mathbf{x^{0}})|\mathbf{\theta}),
\end{eqnarray}
where, 
\begin{equation*} \label{Eq2-5}
\dot{\pi}(\mathbf{\theta}) \propto \pi(\mathbf{\theta}) \exp \left [\frac{g'(s(\mathbf{x^{0}})|\mathbf{\theta})}{g(s(\mathbf{x^{0}})|\mathbf{\theta})} [s(\mathbf{x'}) - s(\mathbf{x^{0}})] + e \right ].
\end{equation*}
\end{proof}

\noindent ABC involves accepting $\mathbf{\theta'} \sim \pi(\mathbf{\theta})$ if  $s(\mathbf{x'}) \in (s(\mathbf{x^{o}}) - \epsilon, s(\mathbf{x^{o}})+\epsilon).$  Note that $\frac{g'(s(\mathbf{x^{0}})|\mathbf{\theta})}{g(s(\mathbf{x^{0}})|\mathbf{\theta})}$ is a function of $\theta$ alone and let it be denoted by $h(\theta).$ Then, for any $|t| = |s(\mathbf{x'}) - s(\mathbf{x^{0}})| \leq \epsilon$ and after ignoring the remainder term $o(t),$  we have an approximation
\begin{eqnarray} \label{Eq2-5*}
\nonumber \pi'(\mathbf{\theta}) & \propto &  \pi(\mathbf{\theta}) \exp \left[ h(\theta) \times t \right],\\
\nonumber & = &  \frac{\pi(\mathbf{\theta}) \exp \left[ h(\theta) \times t \right]}{\int \pi(\mathbf{\theta}) \exp \left[ h(\theta) \times t \right]\, d\theta},\\
& = &  \frac{\pi(\mathbf{\theta}) \exp \left[ h(\theta) \times t \right]}{E_{\pi}\left[ \exp \left[ h(\theta) \times t \right] \right]\,},
\end{eqnarray}

\noindent Thus, for every accepted $\mathbf{\theta'},$ we can find $\pi'(\cdot)$ using \eqref{Eq2-5*}. We define the ABC class of priors for parameter $\epsilon$ as
\begin{equation} \label{Eq2-6}
\Gamma_{\epsilon} = \left\{\pi'(\cdot): \pi'(\mathbf{\theta},t)=\frac{\pi(\mathbf{\theta}) \exp \left[ h(\theta) \times t \right]}{E_{\pi}\left[ \exp \left[ h(\theta) \times t \right] \right]}, 0< |t| \leq \epsilon]\right
 \}.
\end{equation}
Note that since we have ignored the remainder term, $\Gamma_{\epsilon}$ is a class of approximate prior distributions. For $n$-dimensional $\mathbf{\theta}$ and $m$-dimensional $\mathbf{s}=\{s_{1},s_{2},\ldots,s_{m}\},\; m \geq n,$ a similar result can be obtained using the multivariate Taylor series approximation. We state this in Theorem \ref{thm12} without proof.\\

\begin{thm} \label{thm12}
For $n$-dimensional $\mathbf{\theta}, \; m$-dimensional $\mathbf{s}=\{s_{1},s_{2},\ldots,s_{m}\},\; m \geq n$ and for simulated data  $\mathbf{x'},\; \pi(\mathbf{\theta}|\mathbf{x'}) \propto \pi'(\mathbf{\theta}) g(s(\mathbf{x^{0}})|\mathbf{\theta}),$ where,
$$\pi'(\mathbf{\theta}) \propto \pi(\mathbf{\theta}) \exp \left [\sum_{k=1}^{m}\frac{g_{k}}{g(\mathbf{s}(\mathbf{x^{0}}),\mathbf{\theta})}[s_{k}(\mathbf{x'}) - s_{k}(\mathbf{x^{0}})]  + e \right ],$$
where $g_{k} = \frac{\partial \log g(\mathbf{s}(\mathbf{x^{0}},\mathbf{\theta}))}{\partial s_{k}(\mathbf{x^{0}}),\mathbf{\theta})} $ and $e = o[s_{k}(\mathbf{x'}) - s_{k}(\mathbf{x^{0}})].$
 \end{thm}

In the multivariate case, the ABC class of priors is defined as
\begin{equation} \label{Eq2-7}
\Gamma_{\epsilon} = \left \{\pi'(\cdot): \pi'(\mathbf{\theta},t) =  \frac{\pi(\mathbf{\theta}) \exp \left[\sum_{k} h_{k}(\theta) \times t_{k} \right]}{E_{\pi}\left[ \exp \left[ \sum_{k}h_{k}(\theta) \times t_{k} \right] \right]}, 0 < |\mathbf{t}| \leq \mathbf{\epsilon}] \right \},
\end{equation}
where  $ h_{k} (\mathbf{\theta}) = \frac{g_{k}}{g(\mathbf{s}(\mathbf{x^{0}}),\mathbf{\theta})},\,\mathbf{t} = (t_{1},t_{2},\ldots,t_{m})$ and $\mathbf{\epsilon} = (\epsilon_{1},\epsilon_{2},\ldots,\epsilon_{m}).$\\

\subsection{Properties of $\Gamma_{\epsilon}$} \label{props}

\begin{figure}
\includegraphics [scale=0.51]{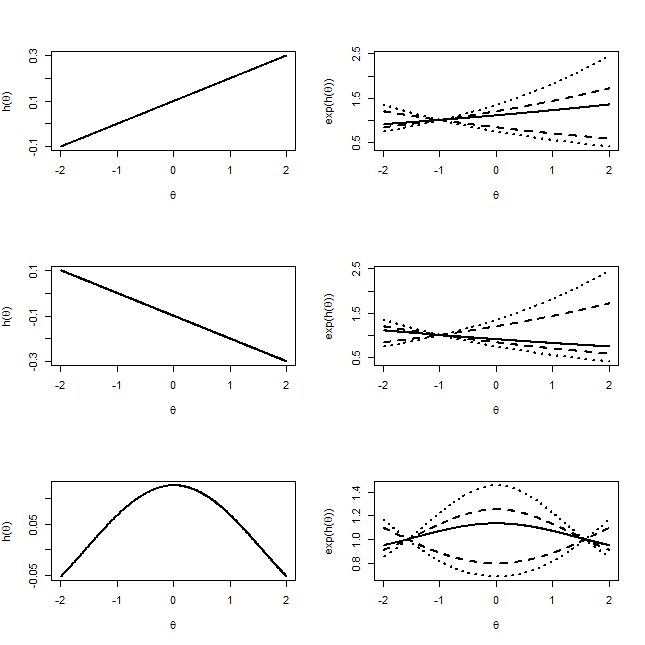}
\caption{Illustrates the monotonicity of $\exp[h(\mathbf{\theta}) \times t ]$ in $t$ for three different types of $h(\mathbf{\theta})$ (\emph{left column})--- \emph{increasing}, \emph{decreasing} and \emph{non-monotonic}. In each case, the \emph{right column} plots the function $\exp[h(\mathbf{\theta})]$  (\emph{solid line}), the functions $\exp[h(\mathbf{\theta}) \times \pm \epsilon ]$ for $\epsilon = \pm 1.8$  (\emph{dashed lines}) and the functions $\exp[h(\mathbf{\theta}) \times \pm \epsilon ]$ for $\epsilon = \pm 3$  (\emph{dotted lines}). }
\label{mono-fig}
\end{figure}

We now look at some of the properties of the ABC class of priors. For the sake of simplicity, we discuss the monotonicity and the topological properties for the one dimensional case only. However, it can be easily seen that they extend to the multivariate case also.\\

Figure \ref{mono-fig} illustrates that (irrespective of whether $h(\mathbf{\theta})$ is monotonic or not), the bands generated by $\exp[h(\mathbf{\theta}) \times \pm \epsilon ]$ are always monotonic in $\epsilon.$  This property is important because it allows us to show that $\Gamma_{\epsilon}$ is a class of prior distributions in the topological sense. As a side note, also observe that while the bands converge to $1$ as $g'(\theta)$ converges to $0,$ this is the derivative of $g$ w.r.t $s(\mathbf{x})$ taking value $0$ for some value/s of $\mathbf{\theta}$ and is merely a mathematical artifact.\\

\noindent \textbf{1. Monotonicity:} The ABC class of priors is monotonic in $\epsilon.$ For any $0<\epsilon_{1}<\epsilon_{2},$ $\pi'\in\Gamma_{\epsilon_1} \Rightarrow \pi'\in\Gamma_{\epsilon_2}$ and therefore $\Gamma_{\epsilon_1} \subset \Gamma_{\epsilon_2}.$ Also, notice that as $\epsilon \downarrow 0,\, \Gamma_{\epsilon} \downarrow \pi.$\\

\noindent \textbf{2. Topological properties:} An implication of monotonicity is that we can establish that $\Gamma_{\epsilon}$ forms a topological space of prior distributions. Let $0< e_{1} < e_{2} < \cdots <\epsilon$ be an increasing sequence of arbitrary length and let $\Gamma_{e_1},\Gamma_{e_2},\cdots$ be the ABC classes of priors defined for the corresponding threshold levels. Since $\pi' \in \Gamma_{e_i} \Rightarrow \pi' \in \Gamma_{e_j},\, \forall i,j: e_{i}<e_{j},$ it can be seen that $\Gamma_{e_1} \subset \Gamma_{e_2} \subset \cdots \subset \Gamma_{\epsilon}.$ This implies that $ \bigcup_{e_i<\epsilon} \Gamma_{e_i} = \Gamma_{\epsilon}.$ Also, for any $e_1<e_2<\cdots<e_n,\, \Gamma_{e_1} \bigcap \Gamma_{e_2} \bigcap \cdots \bigcap \Gamma_{e_n} = \Gamma_{e_1}.$ 

\begin{thm} \label{thm13}
The ABC class of priors $\Gamma_{\epsilon}$ is a topological space.
\end{thm}
\begin{proof}
Given $\epsilon >0,\, \Gamma_{e},$ for any $0\leq e \leq \epsilon$ denotes a subset of $\Gamma_{\epsilon}.$ For $e=0,\, \Gamma_{0} = \emptyset.$ 

\noindent Let $\mathcal{T}_{\Gamma}$ be a collection of subsets of $\Gamma_{\epsilon}$ such that $\Gamma_{e} \in \mathcal{T}_{\Gamma}, \, \forall 0\leq e\leq \epsilon.$ Then  we have that $\emptyset \in \mathcal{T}_{\Gamma}$ and $\Gamma_{\epsilon} \in \mathcal{T}_{\Gamma}.$ Also $\bigcup_{e} \Gamma_{e} \in \mathcal{T}_{\Gamma}$ and, for any $e_1<e_2<\cdots<e_n,\, \Gamma_{e_1} \bigcap \Gamma_{e_2} \bigcap \cdots \bigcap \Gamma_{e_n} \in  \mathcal{T}_{\Gamma}.$ Thus $\mathcal{T}_{\Gamma}$ defines a topology on $\Gamma_{\epsilon}.$
\end{proof}

Although we are in the space of probability distributions, thanks to the way $\pi'$ are defined, the topology on $\pi'$ is the same as that on $\exp[h(\mathbf{\theta})\times t],$ which, in turn, is the same as the topology on the Euclidean space. Thus, $\Gamma_{\epsilon}$ inherits all the nice topological properties of the Euclidean space including compactness.\\

\noindent \textbf{3. Weighted band of priors:} We can show that this class is, in fact, a special case of a more general class of multivariate priors recently proposed, namely,  the \emph{weighted band of priors}. The use of weight functions to modify distributions is not new. Weighted distributions were first introduced by \cite{fisher34} and explicitly defined and studied by \cite{rao63} and \cite{rao85}. \cite{bayari98} introduced non-parametric classes of weight functions to study the robustness of the posterior functionals when the data sampling is modified using a weight function belonging to the said classes. Recently, \cite{ruggeri20} have introduced a new class of multivariate prior distributions based on the weighted distributions called the weighted band of priors. Let $\pi$ be a specific \emph{multivariate total positivity of order $2$} (MTP2) prior belief. Then the weighted band $\Gamma_{w_{1},w_{2},\pi}$ associated with $\pi$ based on $w_{1}$ and $w_{2}$, a decreasing weight function and an increasing weight function, respectively, is defined as 
$$ \Gamma_{w_{1},w_{2},\pi} = \{\pi': \pi_{w_{1}} \leq_{lr} \pi' \leq_{lr} \pi_{w_{2}}\},$$
where, 
$$\pi_{w_{1}} = \frac{\pi(\theta) w_{1}(\theta)}{E_{\pi}[w_{1}(\theta)]} \mbox{ and } \pi_{w_{2}} = \frac{\pi(\theta) w_{2}(\theta)}{E_{\pi}[w_{2}(\theta)]},\, \forall \theta \in \Theta \subseteq \mathbb{R}^{n}. $$

 Let $w_{\mathbf{t}}(\theta) = \exp \left[\sum_{k} h_{k}(\theta) \times t_{k} \right].$ Then, $\pi'(\mathbf{\theta})$ in \eqref{Eq2-7} can be written as 
\begin{equation} \label{Eq2-10}
\pi'(\mathbf{\theta})  =  \frac{\pi(\mathbf{\theta})w_{\mathbf{t}}(\theta)}{E_{\pi}\left[ w_{\mathbf{t}}(\theta)\right]}
\end{equation}
This shows that the ABC class of priors is a special case of the weighted band of priors where  each member of the class takes the above form and where the weight function is defined as above. So far, we have not needed to make the MTP2 assumption on $\pi$ for the ABC class, however, we will need to do so in order to establish the ordering of the members when $\mathbf{\theta}$ is multivariate.\\

\noindent Establishing that $\Gamma_{\epsilon}$ is a topological space is not sufficient to establish the ordering within the space. For example, we are not yet able to say that $\exp[h(\mathbf{\theta}) \times - \epsilon ]$ is (\emph{say}) the \emph{lower bound} and $\exp[h(\mathbf{\theta}) \times + \epsilon ]$  is (\emph{say}) the \emph{upper bound} of $\Gamma_{\epsilon}.$ To do so, we will start by recalling the definition of \emph{likelihood ratio ordering}.\\

\noindent{\textbf{4. Stochastic ordering:}  For one dimensional $\theta,$ let $\pi_{1}(\theta)$ and $\pi_{2}(\theta)$ be two densities that are absolutely continuous w.r.t. the Lebesgue measure. Then, $\pi_{1}(\theta)$ is said to be smaller than $\pi_{2}(\theta)$ in the likelihood ratio order sense --- and denoted $\pi_{1}(\theta) \leq _{lr} \pi_{2}(\theta)$ --- if the ratio $\pi_{2}(\theta)/\pi_{1}(\theta)$ increases over the union of supports of the two densities. This implies that $\pi_{2}(\theta)$ corresponds to a random variable that takes larger values  than the random variable that corresponds to $\pi_{1}(\theta).$ For further details on the likelihood ratio ordering, please refer to \cite{karlin80}, \cite{muller02} and \cite{shaked07}.

\begin{thm}\label{thm_lr}
For univariate $\theta,$ if $h(\theta)$ is monotonic in $\theta$ then the likelihood ratio order can be established inside the $\Gamma_{\epsilon}.$
\end{thm}

\begin{proof} First, note that from \eqref{Eq2-6} and ignoring the normalizing constant, 
\begin{equation}\label{Eq2-8}
\frac{\pi'_{t}(\theta)}{\pi(\theta)} \propto \exp\left[ h(\theta) \times t \right],
\end{equation}
where, subscript $t$ is now used to explicitly denote the value of $t$ being considered for the sake of clarity.\\

\noindent  Let's first consider the case that $h(\theta)$ is monotonic increasing. Then \eqref{Eq2-8} implies that 
\begin{itemize}
\item{$\pi'_{t}(\theta)/\pi(\theta)$ is increasing $\forall t>0, \Rightarrow \, \pi'_{t}(\theta) \geq_{lr}\pi(\theta),$}
\item{$\pi'_{-t}(\theta)/\pi(\theta)$ is decreasing $\forall t>0, \Rightarrow \, \pi'_{-t}(\theta) \leq_{lr}\pi(\theta),$}
\item{$\pi'_{t_{2}}(\theta)/\pi'_{t_{1}}(\theta)$ is increasing $\forall t_{2}>t_{1}>0 \Rightarrow \, \pi'_{t_{2}}(\theta) \geq_{lr}\pi'_{t_{1}}(\theta)$ and substituting $t_{2} = \epsilon$ implies that $ \pi'_{\epsilon}(\theta) \geq_{lr}\pi'_{t}(\theta),\, \forall t<\epsilon,$ and }
\item{$\pi'_{t_{2}}(\theta)/\pi'_{t_{1}}(\theta)$ is decreasing $\forall t_{1}>t_{2} \Rightarrow \, \pi'_{t_{2}}(\theta) \leq_{lr}\pi'_{t_{1}}(\theta)$ and substituting $t_{2} = -\epsilon$ implies that $ \pi'_{-\epsilon}(\theta) \leq_{lr}\pi'_{t}(\theta),\, \forall t>-\epsilon.$ }
\end{itemize}
Thus, when $h(\theta)$ is monotonic increasing we can not only establish that $\pi'_{\epsilon}(\theta)$ and $ \pi'_{-\epsilon}(\theta)$ are the upper and the lower bounds of $\Gamma_{\epsilon},$ respectively, but also establish an ordering for all infinitely many elements within that class.\\

Similarly, it can be seen that when $h(\theta)$ is monotonic decreasing we can establish that $\pi'_{\epsilon}(\theta)$ and $ \pi'_{-\epsilon}(\theta)$ are the lower and the upper bounds of $\Gamma_{\epsilon},$ respectively and also establish an ordering for every element inside the class, only that the ordering will be exactly reversed in this case.
\end{proof}

For multivariate $\mathbf{\theta},$ however, it is not straightforward to show the likelihood ratio ordering and one needs to assume that the prior density $\pi$ is MTP2. This means that for any $\mathbf{\theta}, \mathbf{\theta'} \in \Theta, \, \pi$ satisfies the following condition: $$\pi(\mathbf{\theta})\times \pi(\mathbf{\theta'}) \leq \pi(\mathbf{\theta} \vee \mathbf{\theta'}) \times \pi(\mathbf{\theta} \wedge \mathbf{\theta'}),$$ where  $\vee$ and $\wedge$ denote the component-wise minimum and maximum respectively. Just like for the $lr$ order,  for further details on the MTP2 ordering, please refer to \cite{karlin80}, \cite{muller02} and \cite{shaked07}.\\

\cite{ruggeri20} prove (Lemma 2.8 in their paper) that if $\pi$ is a specific MTP2 prior distribution and $w$ is an increasing (decreasing) weight function, then $\pi \leq_{lr} (\geq_{lr}) \pi_{w}.$ We can use that result to prove the following.

\begin{thm} \label{thm_lr2}
For multivariate $\mathbf{\theta},$ let $\pi$ be a MTP2 prior distribution. If $h_{k}(\mathbf{\theta})$ is increasing in $\theta$ for each $k = 1,\ldots, n,$ or if $h_{k}(\mathbf{\theta})$ is decreasing in $\theta$ for each $k = 1,\ldots, n,$ then the likelihood ratio order can be established inside the $\Gamma_{\epsilon}.$
\end{thm}
\begin{proof} Given the result proved in Lemma 2.8 of \cite{ruggeri20}, all that is needed to be done is to establish the existence of the appropriate increasing and decreasing functions for $\Gamma_{\epsilon}.$ This is done using arguments similar to the proof of Theorem \ref{thm_lr}. First we assume that all $h_{k}'s$ are increasing. Then $\sum_{k} h_{k}(\mathbf{\theta}) \times t_{k}$ is increasing (decreasing) $\forall \mathbf{t} \geq \mathbf{0}$ ( $\forall \mathbf{t} \leq \mathbf{0}$) and that $\mathbf{t} = \mathbf{\epsilon}$ ($\mathbf{t} = \mathbf{-\epsilon}$ ) corresponds to the upper (lower) bound.\\
For the case where all $h_{k}'s$ are decreasing, a similar set of arguments can be made and the ordering will be in exact reverse order.
\end{proof}
%

\noindent \textbf{5. Kolmogorov distance:}
\noindent The Kolmogorov distance between $\pi$ and $\pi'$ (for a given $t$) is 
\begin{eqnarray} \label{Kol1}
\nonumber K(\pi,\pi'_{t}) &=& \sup_{\tau\in \Theta} \left|\int^{\tau} \pi(\mathbf{\theta})\,d\mathbf{\theta} - \int^{\tau}  \pi'_{t}(\mathbf{\theta})\,d\mathbf{\theta}\right| \\
\nonumber & = &  \sup_{\tau\in \Theta} \left|\int^{\tau} \pi(\mathbf{\theta})\,d\mathbf{\theta} - \int^{\tau} \frac{\pi(\mathbf{\theta})w_{t}(\theta)}{E_{\pi}\left[ w_{t}(\theta)^{t}\right]} \,d\mathbf{\theta}\right|\\
 & = &   \sup_{\tau\in \Theta} \left|\int^{\tau} \pi(\mathbf{\theta})\left ( 1 - \frac{w_{t}(\theta)^{t}}{E_{\pi}\left[ w_{t}(\theta)^{t}\right]}\right )\,d\mathbf{\theta}\right|
\end{eqnarray}
Note that $ 0 \le K(\pi,\pi') \le 1$ and $ K(\pi,\pi') \downarrow 0$ as $t \downarrow 0.$ \eqref{Kol1} provides the relationship between the distortion in the data and the distortion in the prior distributions. This could be exploited in two distinct ways. One could find the distance $K(\pi,\pi_{t})$ for a given $t$ to find the distance between prior distributions corresponding to a given $t$ (and thus, for example, the length of $\Gamma_{\epsilon}$ using $t=\pm \epsilon$). But on the other hand, it is also possible to find the maximum value of $t$ such that $ K(\pi,\pi')$ is below a certain upper bound. Therefore, \eqref{Kol1} could be used to elicit $\epsilon$ in a Bayesian analysis, based on the deviation in the prior distributions that we may be prepared to permit. We discuss this further in Section \ref{discuss}. \\

\subsection{Examples} \label{gen-ex}
\textbf{Example 1:} Consider $\mathbf{X} \sim N(\mu,\sigma^{2}),$ with $\mu$ unknown and  $\sigma$ known. Sufficient statistics for $\mu$ is $\bar{X}$. Having observed data $\mathbf{x^{o}},$ we have that
\begin{equation} \label{Ex-1-1}
g(\bar{x}^{0}|\mu) \propto \exp\left [ \frac{-n}{2 \sigma^{2}} (\mu - \bar{\mathbf{x^{0}}})^{2} \right ]  \mbox{  and  } g'(\bar{x}^{0}|\mu)  \propto \frac{n}{\sigma^{2}} (\mu - \bar{\mathbf{x^{0}}}) \exp\left [ \frac{-n}{2 \sigma^{2}} (\mu - \bar{\mathbf{x^{0}}})^{2} \right ].
\end{equation}
Thus, 
\begin{equation} \label{Ex-1}
\pi'(\mu,t) \propto \pi(\mu) \exp \left [ \frac{n}{\sigma^{2}} (\mu - \bar{\mathbf{x^{0}}})\, t  \right ],
\end{equation}
where, $t =  \bar{\mathbf{x^{'}}} -  \bar{\mathbf{x^{0}}}.$  The ABC class of priors in this case would be 
\begin{equation*}
\Gamma^{G}_{\epsilon} = \left\{\pi'(\cdot): \pi'(\mu,t) = \frac{\pi(\mu) \exp \left [\frac{n}{\sigma^{2}} (\mu - \bar{\mathbf{x^{0}}})\, t\right ]}{E_{\pi(\mu)} \left[ \exp \left [\frac{n}{\sigma^{2}} (\mu - \bar{\mathbf{x^{0}}})\, t\right ] \right ] }, t \in [-\epsilon, +\epsilon]\right\}.
\end{equation*}

\textbf{Example 2:} Consider $\mathbf{X} \sim Poisson(\lambda),$ with $\lambda$ unknown. Sufficient statistics for $\lambda$ is $\sum_{i=1}^{n} X_{i}$. Having observed data $\mathbf{x^{o}},$ we have that
$$g(\sum_{i=1}^{n} x^{0}_{i}|\lambda) \propto \exp(-n\lambda)\lambda^{\sum_{i=1}^{n} x^{0}_{i}} \mbox{  and  }  g'(\sum_{i=1}^{n} x^{0}_{i}|\lambda) \propto  \exp(-n\lambda)\lambda^{\sum_{i=1}^{n} x^{0}_{i}} log(\lambda)$$
Thus, 
\begin{equation} \label{Ex-2}
\pi'(\lambda,t) \propto \pi(\lambda) \exp \left [ log(\lambda)\, t  \right ],
\end{equation}
where, $t = \sum_{i=1}^{n} x'_{i} - \sum_{i=1}^{n} x^{0}_{i} .$  The ABC class of priors in this case would be 
\begin{equation*}
\Gamma^{G}_{\epsilon} =\left \{\pi'(\cdot): \pi'(\lambda,t) =\frac{ \pi(\lambda)\lambda^{t}}{E_{\pi(\lambda)}\left[ \lambda^{t} \right ] }, t \in [-\epsilon, +\epsilon] \right \}.
\end{equation*}


\section{Exponential class family of distributions}\label{abc-E}

Consider the special case where $f(\mathbf{x}|\mathbf{\theta})$ belongs to the exponential family and $\pi_{\gamma}(\mathbf{\theta})$ is chosen to be the conjugate prior distribution, where, $\gamma$ denotes the hyper-parameters of the prior distribution. In this case, we can go further and show that, in fact, $\pi'(\cdot) = \pi_{\gamma'}(\cdot).$ That is, $\pi'$ is obtained simply by changing the hyper-parameters $\gamma$. In other words, for every $\mathbf{x'}$ there must exist a $\gamma'$ such that
\begin{equation*}
\pi_{\gamma}(\mathbf{\theta}) f(\mathbf{x'}|\mathbf{\theta} )= \pi_{\gamma'}(\mathbf{\theta}) f(\mathbf{x^{o}}|\mathbf{\theta}).
\end{equation*}

\begin{thm}\label{thm2}
If $f(\mathbf{x}|\mathbf{\theta})$ belongs to the exponential family and $\pi_{\gamma}(\mathbf{\theta})$ is chosen to be the conjugate prior distribution, then for any new data $\mathbf{x'},$ there exists hyper-parameters $\gamma'$ such that 
\begin{equation}\label{premise}
\pi_{\gamma}(\mathbf{\theta}) f(\mathbf{x'}|\mathbf{\theta} )= \pi_{\gamma'}(\mathbf{\theta}) f(\mathbf{x^{o}}|\mathbf{\theta}).
\end{equation}
\end{thm}

\begin{proof}
$f(\mathbf{x}|\mathbf{\theta}) = A(\mathbf{\theta})B(\mathbf{x})\, \exp[C(\mathbf{\theta})S(\mathbf{x})].$ The conjugate prior distribution is of the form $\pi_{\gamma}(\mathbf{\theta}) \propto [A(\mathbf{\theta})]^{k}\,\exp [C(\mathbf{\theta})l],$ where, $\gamma = \{k,l\},$  resulting in a posterior distribution of the form $ \pi(\mathbf{\theta}|\mathbf{x}) \propto  [A(\mathbf{\theta})]^{k+1}\,\exp [C(\mathbf{\theta})(l+S(\mathbf{x}))].$\\

\noindent For a different data $\mathbf{x'}$ such that (WLOG) $S(\mathbf{x'}) - S(\mathbf{x}) = \epsilon,$  we have that
\begin{eqnarray}\label{Eq3-2}
\nonumber \pi_{\gamma}(\mathbf{\theta}) f(\mathbf{x'}|\mathbf{\theta} ) &\propto  &[A(\mathbf{\theta})]^{k+1}\,\exp [C(\mathbf{\theta})(l+S(\mathbf{x'}))]\\
\nonumber  & = & [A(\mathbf{\theta})]^{k+1}\,\exp [C(\mathbf{\theta})((l+\epsilon)+S(\mathbf{x}))]\\
\nonumber & = &   [A(\mathbf{\theta})]^{k}\,\exp [C(\mathbf{\theta})(l+\epsilon)]  [A(\mathbf{\theta})]\,\exp [C(\mathbf{\theta})S(\mathbf{x})],\\
& \propto &  \pi_{\gamma'}(\mathbf{\theta}) f(\mathbf{x}|\mathbf{\theta}) ,
\end{eqnarray}
where,  $\gamma' = \{k,l+\epsilon\}.$  Hence proved.
\end{proof}

\eqref{premise} implies that for every $\epsilon>0$ neighborhood around $s(\mathbf{x^{o}}),$ there exists a $\delta(\epsilon)>0$ neighborhood around $\gamma$ defining the ABC-E class of priors as 
\begin{equation} \label{Eq3-3}
\Gamma^{E}_{\epsilon} = \{\pi_{\gamma'}(\mathbf{\theta}):\gamma' \in (\gamma - \delta(\epsilon),\gamma + \delta(\epsilon))\},
\end{equation}
 where $E$ stands for the Exponential family.\\

\subsection{Comparison with $\Gamma_{\epsilon}$} \label{compare}

\noindent For $t \in [-\epsilon, +\epsilon],\, \pi_{\gamma'}(\mathbf{\theta})$ in \eqref{Eq3-2} can also be written as
\begin{eqnarray}\label{Eq3-4}
\nonumber \pi_{\gamma'}(\mathbf{\theta}) &\propto& [A(\mathbf{\theta})]^{k}\,\exp [C(\mathbf{\theta})(l+ t)]\\
\nonumber &\propto&  [A(\mathbf{\theta})]^{k}\,\exp [C(\mathbf{\theta})\times l] \exp[C(\mathbf{\theta})\times t]\\
\nonumber &\propto& \pi_{\gamma}(\mathbf{\theta}) \exp[C(\mathbf{\theta})\times t]\\
\nonumber & = & \frac{\pi_{\gamma}(\mathbf{\theta}) \exp[C(\mathbf{\theta})\times t]}{\int \pi_{\gamma}(\mathbf{\theta}) \exp[C(\mathbf{\theta})\times t]\, d\theta}\\
 & = & \frac{\pi_{\gamma}(\mathbf{\theta}) \exp[C(\mathbf{\theta})\times t]}{E_{\pi_{\gamma}}\left[ \exp[C(\mathbf{\theta})\times t] \right ]}
\end{eqnarray}

\eqref{Eq2-6} and \eqref{Eq3-4} indicate that the members of both the ABC class of priors $\Gamma_{\epsilon}$ and the ABC-E class of priors $\Gamma_{\epsilon}^{E}$ have the same form but the weights are defined differently. For $\Gamma_{\epsilon},$ they are defined using the function $h(\theta)$, whereas, for  $\Gamma_{\epsilon}^{E}$ they are defined using the function $C(\theta).$ \\

The special case where $\Gamma_{\epsilon}$ is developed for the Exponential family likelihoods, we have that
\begin{equation*}
g(s(\mathbf{x^{0}})|\mathbf{\theta}) = A(\mathbf{\theta})B_{0}(S(\mathbf{x}))\, \exp[C(\mathbf{\theta})S(\mathbf{x})],
\end{equation*}
where we have decomposed the component $B(\mathbf{x}) = \tilde{B}(\mathbf{x})\times B_{0}(S(\mathbf{x}))$ and discarded the component $\tilde{ B}(\mathbf{x})$ since it is not a function of $\mathbf{\theta}$.  Then we have that,
\begin{eqnarray}
\nonumber g'(s(\mathbf{x^{0}})|\mathbf{\theta}) & = & \frac{dg(s(\mathbf{x^{0}})|\mathbf{\theta})}{dS(\mathbf{x})}\\
\nonumber &= & A(\mathbf{\theta})\left[B'_{0}(S(\mathbf{x^{0}}))\, \exp[C(\mathbf{\theta})S(\mathbf{x^{0}})] + B_{0}(S(\mathbf{x^{0}}))C(\mathbf{\theta})\, \exp[C(\mathbf{\theta})S(\mathbf{x^{0}})]\right],
\end{eqnarray}
where, $B'_{0}(S(\mathbf{x})) = \frac{dB_{0}(S(\mathbf{x}))}{dS(\mathbf{x})}$ and $\frac{dS(\mathbf{x})}{dS(\mathbf{x})} = 1.$
 Therefore,
\begin{equation} \label{Eq3-5}
h(\mathbf{\theta}) = \frac{g'(s(\mathbf{x^{0}})|\mathbf{\theta})}{g(s(\mathbf{x^{0}})|\mathbf{\theta})} = \frac{B'_{0}(S(\mathbf{x^{0}})) }{B_{0}(S(\mathbf{x^{0}})) } + C(\mathbf{\theta}).
\end{equation}

\eqref{Eq3-4} and \eqref{Eq3-5} show that for likelihoods that belong to the Exponential family, the members of both $\Gamma_{\epsilon}$ and $\Gamma_{\epsilon}^{E}$ will have the same parametric form, but members of $\Gamma_{\epsilon}$ will contain an extra term $\exp\left[\frac{B'_{0}(S(\mathbf{x^{0}}))}{B_{0}(S(\mathbf{x^{0}})) }\right].$

\subsection{Examples}

\textbf{Example 1 (contd.):} For $\mathbf{X} \sim N(\mu,\sigma^{2}),$ with $\mu$ unknown and  $\sigma$ known, conjugate prior for $\mu$ is $N(m,s^{2})$. The posterior distribution is $N(m',s'^{2}),$ where $m' = w_{1}m+w_{2} \mathbf{\bar{x}^{0}},$ where,  $$\frac{1}{s'^{2}} = \frac{1}{s^{2}}+\frac{n}{\sigma^{2}} \mbox{ and } w_{1} = \frac{s'^{2}}{s^{2}},\; w_{2} = \frac{n s'^{2}}{\sigma^{2}}.$$\\

Accepting $\mathbf{x'}$ such that $\mathbf{\bar{x}'} = \mathbf{\bar{x}^{0}} + \epsilon$ implies that the posterior mean now becomes $m'(\epsilon) = m' + w_{2}\epsilon,$ which can also be obtained as $m'(\epsilon) = w_{1} (m+\frac{w_{2}}{w_{1}}\epsilon) + w_{2}\mathbf{\bar{x}^{0}}.$ That is, the same posterior mean $m'(\epsilon)$ can be obtained either by accepting a distorted data $\mathbf{x'}$ or by using the distorted prior with mean $m+\frac{w_{2}}{w_{1}}\epsilon.$ The ABC class of priors in this case will be 
\begin{equation*}
\Gamma^{E}_{\epsilon}=\{ N(m'',s^{2}): m'' \in (m -\frac{w_{2}}{w_{1}}\epsilon, m+\frac{w_{2}}{w_{1}}\epsilon)\}.
\end{equation*}
Note that, when $\sigma^{2}$ is known, the most natural choice is $s^{2} = \sigma^{2}/n,$ and in this case, $\frac{w_{2}}{w_{1}}=1,$ resulting in 
\begin{equation*}
\Gamma^{E}_{\epsilon}=\{ N(m'',s^{2}): m'' \in (m -\epsilon, m+\epsilon)\}.
\end{equation*}

\noindent Now, to derive $\Gamma_{\epsilon}$ from $\Gamma^{E}_{\epsilon},$ note that, in this case, the joint likelihood can be factorised as
\begin{eqnarray}
\nonumber  f(\mathbf{x}|\mathbf{\theta} ) & = & (2\pi\sigma^{2})^{-n/2} \times \exp\left[\frac{-1}{2\sigma^{2}}\sum_{i=1}^{n} (x_{i}-\bar{x})^{2}\right]\times \exp\left[\frac{-n}{2\sigma^{2}}(\mu-\bar{x})^{2}\right]\\
\nonumber & = &  \tilde{B}(\mathbf{x})\times g(\mathbf{\bar{x}}|\mu),
\end{eqnarray}
where $$\tilde{B}(\mathbf{x}) = (2\pi\sigma^{2})^{-n/2} \times \exp\left[\frac{-1}{2\sigma^{2}}\sum_{i=1}^{n} (x_{i}-\bar{x})^{2})\right]$$
and $g(\bar{x}|\mu)$ is as defined in \eqref{Ex-1-1}. Therefore, we have that 

$$A(\mu) = \exp\left[\frac{-n\mu^{2}}{2\sigma^{2}}\right],\, B^{0}(\bar{x}) =  \exp\left[\frac{-n\bar{x}^{2}}{2\sigma^{2}}\right],\,C(\mu) = \frac{n\mu}{\sigma^{2}}\, \mbox{ and}\, S(\mathbf{x}) = \bar{x}.$$ Then, $$ \frac{B'_{0}(\bar{x})) }{B_{0}(\bar{x}) } =   \frac{-n\bar{x}}{\sigma^{2}}.$$  Thus, for $\mathbf{x^{0}},$ using \eqref{Eq3-5} we get $$h(\mu) =  \exp \left [ \frac{n}{\sigma^{2}} (\mu - \mathbf{\bar{x}^{0}})\right ],$$ same as that obtained in Section \ref{gen-ex} using the Taylor series approximations.\\ 

\textbf{Example 2 (contd.):}For $\mathbf{X} \sim Poisson(\lambda),$ with $\lambda$ unknown, conjugate prior for $\lambda$ is $Gamma(r,v),$ where $r$ is the shape and $v$ is the rate parameter. The posterior distribution is $Gamma(r',v'),$ where $r' = r + \sum_{i=1}^{n} x^{0}_{i}$ and $v' = v+n.$ Accepting $\mathbf{x'}$ such that $\sum_{i=1}^{n} x'_{i} = \sum_{i=1}^{n} x^{0}_{i} + \epsilon$ implies that the posterior now becomes $Gamma(r'(\epsilon),v'),$ where $r'(\epsilon) = r' + \epsilon = (r+\epsilon) + \sum_{i=1}^{n} x^{0}_{i},$ which can also be obtained by using $Gamma(r+\epsilon,v)$ as the prior distribution.  That is, the same posterior can be obtained either by accepting a distorted data $\mathbf{x'}$ or by using the distorted prior with parameter $r+\epsilon.$ The ABC class of priors in this case will be 
\begin{equation*}
\Gamma^{E}_{\epsilon}=\{ Gamma(r'',v): r'' \in (r-\epsilon, r+ \epsilon)\}.
\end{equation*}

These two examples also illustrate that $\Gamma_{\epsilon}$ and $\Gamma_{\epsilon}^{E}$ are distinct classes even for models where the likelihood is from the Exponential family of distributions and the prior distribution is the corresponding conjugate prior distribution. However, the difference between their upper and lower  limits may not necessarily be significant.  For example, we consider a simulated data for the Normal distribution example. Let us assume that for $n=100, \, X_{1},\ldots,X_{n} \sim N(\mu, 2).$ Let $\pi(\mu)$ be $N(10,2/\sqrt{n})$ and the observed sufficient statistic $ \bar{x}^{0} =9.975.$ The two classes obtained in this case are shown in Figure \ref{normal_fig}. It shows that in this case, $\Gamma_{\epsilon}$ and $\Gamma_{\epsilon}^{E}$ appear to be almost exactly the same. This is because in this case, $$ \frac{B'_{0}(\bar{x})) }{B_{0}(\bar{x}) } =   \frac{-n\bar{x}}{\sigma^{2}} = - 249.375$$
and $\exp[-249.375] \approx 0.$ Of course, this may not necessarily be the case if, say, $n\bar{x}/\sigma^{2}$ was much closer to $0.$

\begin{figure} [ht] 
\includegraphics [scale=0.2]{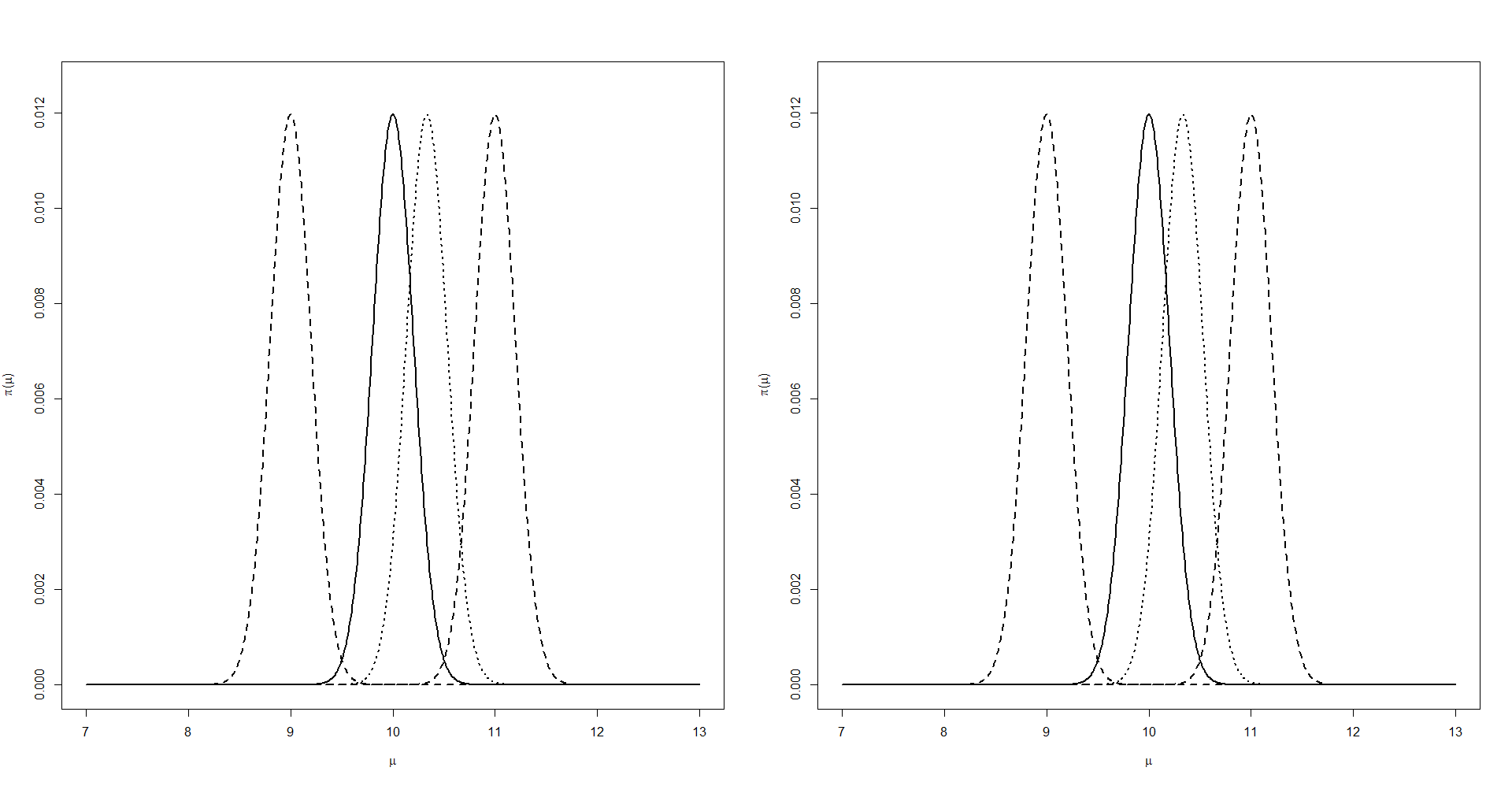}
\caption{The original prior (solid), the upper and lower (dashed) bounds and an internal member (dotted) for the Normal distribution example using the [\emph{left}] ABC class of priors $\Gamma_{\epsilon}$ and [\emph{right}] the ABC-E class of priors $\Gamma^{E}_{\epsilon}$, for $\epsilon =1.$}
\label{normal_fig}
\end{figure}

\subsection{Properties of $\Gamma^{E}_{\epsilon}$}

As discussed in Section \ref{compare} the members of both $\Gamma_{\epsilon}$ and $\Gamma_{\epsilon}^{E}$ share a similar structure. Due to this similarity, it can be seen that the properties of  the $\Gamma_{\epsilon}$ class of priors discussed in Section \ref{props} are also applicable to  the  $\Gamma_{\epsilon}^{E}$ class of priors. Specifically, for the stochastic ordering this is subject to $C(\mathbf{\theta})$ satisfying the conditions that were placed on $h(\mathbf{\theta})$ and once again, $\pi$ being a $MTP2$ distribution for multivariate $\mathbf{\theta}.$


\section{Computational Aspects} \label{algo}

Here, we want to provide Monte Carlo algorithms that enable sampling from the prior distributions that belong to the classes of priors defined in this paper. Further, we will show that the posterior distribution obtained using the ABC method can also be obtained by sampling from the posterior distributions $ \pi(\mathbf{\theta}|\mathbf{x'})$ corresponding to the prior distributions $\pi'(\theta)$ that belong to the ABC classes of priors. 

\subsection{Sampling from the Prior Distributions}

First we address the question of how does one sample from a distribution that belongs to the ABC classes of prior distributions. The answer to this question is only likely to be non-trivial for the ABC class of prior distributions $\Gamma_{\epsilon}.$ This is because the members of the ABC-E class of prior distributions $\Gamma_{\epsilon}^{E}$ have the same form as the original conjugate prior distribution $\pi$ and should therefore be easy to be sampled from. Observe regardless that, given $t \in [-\epsilon, +\epsilon],$ 
\begin{equation} \label{Imp}
\frac{\pi'(\theta)}{\pi(\theta)} = \frac{w(\theta)^{t}}{E_{\pi}[w(\theta)^{t}]},
\end{equation}
is true for both $\Gamma_{\epsilon}$ as well as $\Gamma_{\epsilon}^{E}$.  \eqref{Imp} can be used to sample from $\pi'$ using the importance sampling approach.\\

\noindent \underline{\textbf{A.I. Importance sampling algorithm to sample from $\pi'(\theta,t) \in \Gamma_{\epsilon}$ or $\Gamma^{E}_{\epsilon}.$}}\\

\noindent Given $S(\cdot), \mathbf{x}^{0}$ and $0<t<\epsilon,$
\begin{description}
\item{Step I: For large $N,$ sample $\theta_{1},\ldots,\theta_{N} \sim \pi(\theta)$.}
\item{Step II: For each $\theta_{i},\, i=1,\ldots, N,$ compute the importance weight
$$\omega_{i} = \frac{w(\theta_{i})^{t}}{E_{\pi}[w(\theta_{i})^{t}]}.$$}
\item{Step III: $\{(\theta_{1},\omega_{1}),\ldots, (\theta_{N},\omega_{N})\}$ is a weighted sample from $\pi'(\theta,t).$}
\end{description}

%
%
%

\subsection{Sampling from the Posterior Distributions $ \pi(\mathbf{\theta}|\mathbf{x'})$}

Next, we consider sampling from the posterior distributions $ \pi(\mathbf{\theta}|\mathbf{x'})$ that correspond to the prior distributions $\pi'(\theta).$  From Theorems \ref{thm1} and \ref{thm2} we have that
\begin{eqnarray} \label{imp-1}
\nonumber \pi(\mathbf{\theta}|\mathbf{x'}) &\propto& \pi'(\mathbf{\theta}) g(s(\mathbf{x^{0}})|\mathbf{\theta}),\\
\nonumber &\propto& \pi(\mathbf{\theta}) \exp \left[ h(\theta) \times t \right] g(s(\mathbf{x^{0}})|\mathbf{\theta}),\\
& = & \pi(\mathbf{\theta}) \frac{w_{\mathbf{t}}(\theta)g(s(\mathbf{x^{0}})|\mathbf{\theta})}{E_{\pi}\left[ w_{\mathbf{t}}(\theta) g(s(\mathbf{x^{0}})|\mathbf{\theta})\right].}
\end{eqnarray}

Thus, for each $t \in (-\epsilon,+\epsilon),$ we can use importance sampling to sample from the posterior  $\pi(\mathbf{\theta}|\mathbf{x'})$.\\

\noindent \underline{\textbf{A.II. Importance sampling algorithm to sample from $ \pi(\mathbf{\theta}|\mathbf{x'})$}}\\

\noindent Given $S(\cdot), \mathbf{x}^{0}$ and $0<t<\epsilon,$
\begin{description}
\item{Step I: For large $N,$ sample $\theta_{1},\ldots,\theta_{N} \sim \pi(\theta).$}
\item{Step II: For each $\theta_{i}, \, i=1,\ldots, N,$ compute the importance weight
$$\omega_{i} =  \frac{w_{\mathbf{t}}(\theta)g(s(\mathbf{x^{0}})|\mathbf{\theta})}{E_{\pi}\left[ w_{\mathbf{t}}(\theta) g(s(\mathbf{x^{0}})|\mathbf{\theta})\right].}$$}
\item{Step III: $\{(\theta_{1},\omega_{1}),\ldots, (\theta_{N},\omega_{N})\}$ is a weighted sample from $\pi(\mathbf{\theta}|\mathbf{x'}).$}
\end{description}

Note that, here we are able to sample from $ \pi(\mathbf{\theta}|\mathbf{x'})$ without having to simulate $\mathbf{x'}$ and by simply using the modified prior $\pi'$ instead. Again, from Theorem \ref{thm2} when the likelihood is from the Exponential family and a conjugate prior is used, sampling from $\pi(\mathbf{\theta}|\mathbf{x'})$ is trivial since it will be a closed form distribution that can be directly sampled from.\\

\subsection{Sampling from the Posterior Distribution $ \pi(\mathbf{\theta}|\mathbf{x^{0}})$}

The rationale behind the ABC methods is to accept a value $\mathbf{\theta} \sim \pi(\mathbf{\theta})$ if $\mathbf{x'}$ simulated using it is considered to be \emph{close enough} so that $\mathbf{x'} \approx \mathbf{x^{0}}$. However, $\mathbf{x'} \approx \mathbf{x^{0}} \Rightarrow \pi(\mathbf{\theta}|\mathbf{x'}) \approx \pi(\mathbf{\theta}|\mathbf{x^{0}})$ (assuming that the likelihood is smooth). That is, sampling from $ \pi(\mathbf{\theta}|\mathbf{x^{0}})$ can be approximated by sampling from  $\pi(\mathbf{\theta}|\mathbf{x'})$ instead. Thus, if we sample $\mathbf{\theta}$ from $ \pi(\mathbf{\theta}|\mathbf{x'})$ for every value $\mathbf{x'}$ that was close enough, then the resulting sample could be considered to be from $ \pi(\mathbf{\theta}|\mathbf{x^{0}}).$ The duality established in Sections \ref{abc-G} and \ref{abc-E} implies that $ \pi(\mathbf{\theta}|\mathbf{x'})$ can be accessed by using the corresponding $\pi'$ instead. Therefore, we should be able to sample from $ \pi(\mathbf{\theta}|\mathbf{x^{0}})$ by sampling from $ \pi(\mathbf{\theta}|\mathbf{x'})$ obtained using $\pi'$ for each $\mathbf{x'}$ that was close enough.

\noindent \underline{\textbf{A.III. Importance sampling algorithm to sample from $ \pi(\mathbf{\theta}|\mathbf{x^{0}})$}}\\

\noindent Given $S(\cdot), \mathbf{x}^{0},\,\epsilon > 0,$ and $N$ and $m$ sufficiently large,
\begin{description}
\item{Step I: Sample $t_{1},\ldots.t_{N} \sim U(-\epsilon,+\epsilon).$}
\item{Step II: For each $t_{i},$ obtain $\{(\theta_{i1},\omega_{i1}),\ldots, (\theta_{im},\omega_{im})\}$ from the corresponding $\pi(\mathbf{\theta}|\mathbf{x'})$ using Algorithm A.II.}
\item{Step III:Normalise the $N \times m$ weights $\omega_{i1},\ldots,\omega_{im}$ for $i=1,\ldots,N.$}
\item{Step IV: $\{(\theta_{i1},\omega_{i1}),\ldots, (\theta_{im},\omega_{im})\},$ for $i=1,\ldots,N,$  is a weighted sample from $\pi(\mathbf{\theta}|\mathbf{x^{0}}).$}
\end{description}

Again, note that,  for any $ \pi(\mathbf{\theta}|\mathbf{x'})$ that corresponds to a prior distribution in $\Gamma_{\epsilon}^{E},$ we could sample directly from it in Step II without having to use Algorithm A.II.

We continue the Normal example in Section 3.2 and compare the true posterior with the posterior obtained using a standard ABC algorithm, the posterior obtained using Algorithm A.III and the posterior obtained by directly sampling from the posteriors for each of the priors in $\Gamma_{\epsilon}^{E}.$  These posteriors are plotted in Figure \ref{abc-post}.

\begin{figure} [ht]
\includegraphics [scale=0.45]{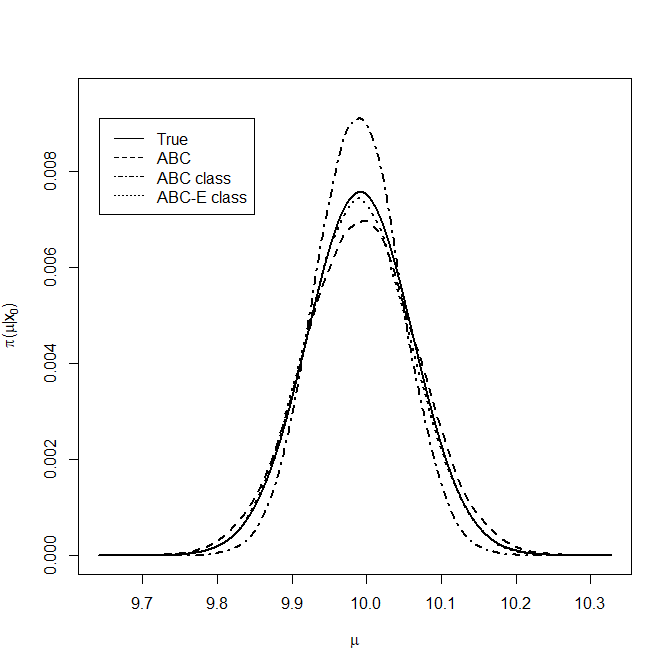}
\caption{The true posterior (solid), the posterior obtained using a standard ABC algorithm (dashed), the posterior obtained by sampling from the posteriors obtained using the ABC class of priors - Algorithm A.V - (dot-dashed) and  the posterior obtained by direct sampling from the posteriors obtained using the ABC-E class of priors (dotted) for the Normal distribution example.}
 \label{abc-post}
\end{figure}

\section{Generalizations using Likelihood} \label{general}

The ABC classes of priors defined so far are applicable for problems where the sufficient statistic is available. However, for many complex problems, the sufficient statistic is often not available or not known. For some of these problems, the likelihood function may be known, but for others the exact likelihood function may not be available either. It is possible to generalize the ideas considered in Sections \ref{abc-G} and \ref{abc-E} in two different ways. Firstly, we can define the class of priors where the sufficient statistic is not available but we do have the likelihood $f(\cdot|\mathbf{\theta}).$ Secondly, we can also define a class of priors where the likelihood is not available in a closed form and therefore an approximate likelihood function $\tilde{f}(\cdot|\mathbf{\theta})$ is to be used instead.

\subsection{Likelihood is available} \label{general-1}

Consider the case where the sufficient statistics $S(\cdot)$ is not available but the likelihood $f(\cdot|\mathbf{\theta})$ is. Suppose we observe data $\mathbf{x^{0}},$ then the posterior distribution will be given by
\begin{equation}
\nonumber \pi(\mathbf{\theta}|\mathbf{x^{0}}) = \frac{\pi(\mathbf{\theta}) f(\mathbf{x^{0}}|\mathbf{\theta})}{m_{0}(\mathbf{x^{0}})},
\end{equation}
where $m_{0}(\mathbf{x^{0}}) = \int \pi(\mathbf{\theta}) f(\mathbf{x^{0}}|\mathbf{\theta}) \,d\theta.$ While $\mathbf{x^{0}}$ has been observed, we believe that we may as well expect to observe a slightly different data $\mathbf{x'}.$ Now, if we were to observe a data $\mathbf{x'}$ instead, we should still be able to arrive at the same posterior distribution by changing the prior distribution accordingly. Therefore we have
\begin{equation}
\nonumber \pi(\mathbf{\theta}|\mathbf{x^{0}}) = \frac{\pi'(\mathbf{\theta}) f(\mathbf{x'}|\mathbf{\theta})}{m_{1}(\mathbf{x'})},
\end{equation}
where $m_{1}(\mathbf{x'}) = \int \pi'(\mathbf{\theta}) f(\mathbf{x'}|\mathbf{\theta}) \,d\theta.$ This gives us that
\begin{equation} \label{Eq6-1}
\pi'(\mathbf{\theta}) = \frac{f(\mathbf{x^{0}}|\mathbf{\theta})}{f(\mathbf{x'}|\mathbf{\theta})}\frac{m_{1}(\mathbf{x'})}{m_{0}(\mathbf{x^{0}})}\, \pi(\mathbf{\theta}).
\end{equation}
\eqref{Eq6-1} provides the general relationship between distortion in the data and distortion in the prior distributions. This relationship can be explored to generate a new general class of prior distributions that can be used when there is a case to be made for uncertainty in the observed data and where a slightly different data was considered to be admissible under the same posterior distribution.\\

In the case where the likelihood comes from the Exponential family, but the prior distribution may not be a conjugate prior, we have that
$$ \frac{f(\mathbf{x^{0}}|\mathbf{\theta})}{f(\mathbf{x'}|\mathbf{\theta})} = \frac{B(\mathbf{x^{0}})}{B(\mathbf{x'})}\, \exp\left[ C(\mathbf{\theta}) (s(\mathbf{x^{0}}) - s(\mathbf{x'}))\right]\,\mbox{and}\,  \frac{m_{1}(\mathbf{x'})}{m_{0}(\mathbf{x^{0}})} = \frac{B(\mathbf{x'}) \times I'}{B(\mathbf{x^{0}}) \times I^{0}},$$
where $I^{0}$ and $I'$ are normalizing constants after factoring $B$ out. This implies that we have
\begin{equation} \label{Eq6-2}
\pi'(\mathbf{\theta}) \propto \pi(\mathbf{\theta})  \exp\left[ C(\mathbf{\theta}) (s(\mathbf{x^{0}}) - s(\mathbf{x'}))\right],
\end{equation}
the same as the one obtained for the class $\Gamma_{\epsilon}^{E}.$ \eqref{Eq6-2} provides a powerful result because it shows that the $\Gamma_{\epsilon}^{E}$ class is applicable when the likelihood belongs to the Exponential family of distributions, irrespective of whether conjugate priors were used or not.\\

When the likelihood does not come from the Exponential family of distributions, we can once again use the Taylor series approximation to derive a class of priors. We have that
$$\frac{f(\mathbf{x^{0}}|\mathbf{\theta})}{f(\mathbf{x'}|\mathbf{\theta})} = \exp\left[ \log(f(\mathbf{x^{0}}|\mathbf{\theta})) - \log(f(\mathbf{x'}|\mathbf{\theta})) \right].$$
Approximating $ \log(f(\mathbf{x'}|\mathbf{\theta}))$ using the first order Taylor's approximation around $ \log(f(\mathbf{x^{0}}|\mathbf{\theta}))$ we get that
\begin{equation} \label{Eq6-3}
\frac{f(\mathbf{x^{0}}|\mathbf{\theta})}{f(\mathbf{x'}|\mathbf{\theta})} = \exp\left[ -\frac{\sum_{i=1}^{N} f_{i}(x^{0}_{i}|\mathbf{\theta}) (x'_{i} - x^{0}_{i})}{f(\mathbf{x^{0}}|\mathbf{\theta})}  + e \right],
\end{equation}
where, $f_{i} = \frac{d\,\log f(x_{i}|\mathbf{\theta})}{dx_{i}}$ for data of size $N.$ Then using \eqref{Eq6-1} we get that
\begin{equation} \label{Eq6-4}
\pi'(\mathbf{\theta}) \propto   \exp\left[  -\frac{\sum_{i=1}^{N} f_{i}(x^{0}_{i}|\mathbf{\theta}) (x'_{i} - x^{0}_{i})}{f(\mathbf{x^{0}}|\mathbf{\theta})} + e \right] \, \pi(\mathbf{\theta}).
\end{equation}
Having observed $\mathbf{x^{0}},\,\frac{ f_{i}(\mathbf{x^{0}}|\mathbf{\theta})}{f(\mathbf{x^{0}}|\mathbf{\theta})}$ can be denoted using a function of $\mathbf{\theta}$ alone, say $k_{i}(\mathbf{\theta}).$  Ignoring the first order residual term, we get the new general class of approximate prior distributions
\begin{equation} \label{Eq6-5}
\Gamma_{\epsilon}^{G} = \left\{\pi'(\cdot): \pi'(\mathbf{\theta},t)=\frac{\pi(\mathbf{\theta}) \exp \left[-\sum_{i=1}^{N} k_{i}(\theta) \times t_{i} \right]}{E_{\pi}\left[ \exp \left[-\sum_{i=1}^{N}k_{i}(\theta) \times t_{i} \right] \right]}, t_{i} \in [-\epsilon, +\epsilon]\right
 \},
\end{equation}
where the superscript $G$ stands for general. Notice that the negative sign for $\sum k_{i}(\theta)$ is only because here we have wanted to arrive at the posterior distribution $\pi(\mathbf{\theta}|\mathbf{x^{0}})$ using data $\mathbf{x'}$ unlike for the class $\Gamma_{\epsilon}$ where we wanted to arrive at the posterior distribution $\pi(\mathbf{\theta}|\mathbf{x'})$  using data $\mathbf{x^{0}}.$  \eqref{Eq6-5} is very powerful because it provides a class of priors that is applicable to any likelihood and any prior distribution even when the sufficient statistics is not available, as is the case in many complex real life problems that Bayesian inference attempts to solve. 

\subsection{Likelihood is not available}

The motivation for the development of ABC approaches comes from examples (\cite{tavare97}, \cite{pritchard99}) where the likelihood was intractable. ABC methods get around the need to compute the likelihood altogether by relying on simulating data instead.  Indeed, the problem of intractable likelihoods predates the development of ABC methods and several other approaches have been proposed over time that involve approximating the likelihood using a tractable function. These include, for example, pseudo-likelihoods (\cite{besag75}), variational Bayes (\cite{jordan99}), expectation-propagation (\cite{minka01}), the Integrated Nested Laplace Approximation (INLA) (\cite{rue09}) and synthetic likelihood (\cite{wood10}).\\

It is possible to look at such approximations instead as an implicit exercise in prior robustness. Again, as in Section \ref{general-1}, suppose that we observe data $\mathbf{x^{0}},$ then the posterior distribution will be given by
\begin{equation} \label{Eq6-10}
 \pi(\mathbf{\theta}|\mathbf{x^{0}}) = \frac{\pi(\mathbf{\theta}) f(\mathbf{x^{0}}|\mathbf{\theta})}{m(\mathbf{x^{0}})},
\end{equation}
where  $m(\mathbf{x^{0}}) = \int \pi(\mathbf{\theta}) f(\mathbf{x^{0}}|\mathbf{\theta}) \,d\theta.$   However, assume now, that the likelihood is intractable and instead an approximation $\tilde{f}(\cdot|\mathbf{\theta})$ is to be used instead. This will then result in an approximate posterior distribution given by 
\begin{equation}
\nonumber \tilde{\pi}(\mathbf{\theta}|\mathbf{x^{0}}) = \frac{\pi(\mathbf{\theta}) \tilde{f}(\mathbf{x^{0}}|\mathbf{\theta})}{\tilde{m}(\mathbf{x^{0}})},
\end{equation}
where  $\tilde{m}(\mathbf{x^{0}}) = \int \pi(\mathbf{\theta}) \tilde{f}(\mathbf{x^{0}}|\mathbf{\theta}) \,d\theta.$  It may be possible though to obtain the true posterior distribution when using $\tilde{f}(\cdot|\mathbf{\theta})$ by using a different prior $\pi'(\cdot)$ instead. We then have,
\begin{equation} \label{Eq6-12}
 \pi(\mathbf{\theta}|\mathbf{x^{0}}) = \frac{\pi'(\mathbf{\theta}) \tilde{f}(\mathbf{x^{0}}|\mathbf{\theta})}{m'(\mathbf{x^{0}})},
\end{equation}
where  $m'(\mathbf{x^{0}}) = \int \pi'(\mathbf{\theta}) \tilde{f}(\mathbf{x^{0}}|\mathbf{\theta}) \,d\theta.$  Ignoring the normalizing constants, \eqref{Eq6-10} and \eqref{Eq6-12} give us that
\begin{equation} \label{Eq6-13}
\pi'(\mathbf{\theta}) \propto \frac{f(\mathbf{x^{0}}|\mathbf{\theta})}{\tilde{f}(\mathbf{x^{0}}|\mathbf{\theta})}\, \pi(\mathbf{\theta}), 
\end{equation}
and of course, if in addition the data is simulated as in ABC then, 
\begin{equation} \label{Eq6-14}
\pi'(\mathbf{\theta}) \propto \frac{f(\mathbf{x^{0}}|\mathbf{\theta})}{\tilde{f}(\mathbf{x'}|\mathbf{\theta})}\, \pi(\mathbf{\theta}). 
\end{equation}
For observed $\mathbf{x^{0}}, \, \frac{f(\mathbf{x^{0}}|\mathbf{\theta})}{\tilde{f}(\mathbf{x'}|\mathbf{\theta})}$ is a function of $\mathbf{\theta}$ alone and can be denoted by, say, $h(\mathbf{\theta}).$   We can thus define a new general class of priors - we call it the Approximate Bayesian (AB) class of priors - that is applicable when the likelihood is approximated in a Bayesian inference problem. This is given by

\begin{equation} \label{Eq6-15}
\Gamma^{A} = \left\{\pi'(\cdot): \pi'(\mathbf{\theta}) =  \frac{h(\mathbf{\theta})  \pi(\mathbf{\theta})}{E_{\pi}(\mathbf{\theta})} \right \}, 
\end{equation}
where, note that the parameter $\epsilon$ can not  be incorporated until we are able to define a distance function to compare $f$ and $\tilde{f}.$  If $f$ is intractable, it may not be possible to derive the exact expression for $\pi'.$ But it is clear that a Bayesian inference problem where the likelihood is approximated can also be viewed as an implicit exercise in prior robustness instead of an exercise in  approximate inference.


\section{Discussion} \label{discuss}

The primary contribution of this paper is theoretical. It establishes the duality between approximate Bayesian methods and prior robustness. We show that implementing ABC methods is equivalent to an exercise in prior robustness using the corresponding ABC classes of priors. We also sketch how a similar class of priors can be defined for other approximate Bayesian methods where the likelihood is approximated. The discourse so far has been largely mathematical in nature. We will now discuss interpretation, application, elicitation and computational aspects of these classes of priors.


The duality established here shows that it is possible to absorb any possible distortions/ approximations to the likelihood into the prior distribution without changing the resulting posterior distribution. Where the likelihood was available and a slightly different data was considered admissible, this duality manifests itself through \eqref{Eq6-1}.  On the other hand, where the likelihood has been approximated, the duality is generalized through \eqref{Eq6-13} or \eqref{Eq6-14}.  Specifically for the ABC method, while we initially define the class $\Gamma_{\epsilon}^{E}$ assuming that the sufficient statistics exists and that conjugate prior distributions are used, \eqref{Eq6-2} shows that the class holds even after relaxing both of these assumptions. Further, the class $\Gamma_{\epsilon}^{G}$ obtained by relaxing the sufficient statistics assumption has only minor differences compared to the class $\Gamma_{\epsilon}$ defined assuming that the sufficient statistic exists. These generalizations enable the 'Approximate' aspect of ABC and other approximation methods to be instead viewed as an implicit exercise in prior robustness.

An important counter implication of this work is that it suggests that it will indeed be very difficult to implement meaningful prior robustness analysis on ABC methods. In a typical prior robustness analysis, if one uses an alternative prior $\pi'(\mathbf{\theta})$ then one would expect to get a different posterior $\pi'(\mathbf{\theta}|\mathbf{x^{0}}) \propto f(\mathbf{x^{0}}|\mathbf{\theta}) \pi'(\mathbf{\theta}).$ However, we have shown here that  $\pi'(\mathbf{\theta}|\mathbf{x^{0}}) \propto \pi(\mathbf{\theta}|\mathbf{x'}),$ for some $\mathbf{x'}$ and as long as that $\mathbf{x'}$ is close enough to $\mathbf{x^{0}},$ it has already been accounted for in the ABC posterior $\pi(\mathbf{\theta}|\mathbf{x^{0}}).$ This is precisely what it means to say that the prior robustness analysis is implicit in an ABC method.

The practical application of ABC classes of priors would lie outside of the ABC methods. Where the observed data is rather limited and there are grounds for concern that it represents only a subset of the values considered to be likely, the ABC classes of priors along with the duality shown here could provide a means to model the uncertainty in the data within the Bayesian framework via the prior robustness analysis. Consider, for example, a researcher trying to model count data $y^{0}_{1},\ldots, y^{0}_{n}$ using a Poisson regression model $$\log(E(Y_{i})) = \mathbf{\beta}'\mathbf{X}_{i}.$$ Suppose that $Y$ are observed sparsely and the researcher thinks that they could have quite easily observed a slightly different data $y'_{i},$ such that $y'_{i} \in (y^{0}_{i} \pm \epsilon_{i})$ and the researcher is able to elicit $\epsilon_{i}$ for each $i.$ In such a case, the researcher could derive the ABC class of priors $\Gamma_{\epsilon}$ and determine that for them

\begin{equation*} 
\Gamma_{\epsilon} = \left \{\pi'(\mathbf{\beta}): \pi'(\mathbf{\beta},t) =  \frac{\pi(\mathbf{\beta}) \exp \left[\sum_{i} \mathbf{\beta}'\mathbf{X}_{i} \times t_{i} \right]}{E_{\pi}\left[ \exp \left[\sum_{i} \mathbf{\beta}'\mathbf{X}_{i} \times t_{i} \right]\right ]}, 0 < |\mathbf{t}| \leq \mathbf{\epsilon}\right \},
\end{equation*}
where $\mathbf{t} = (t_{1},t_{2},\ldots,t_{m})$ and $\mathbf{\epsilon} = (\epsilon_{1},\epsilon_{2},\ldots,\epsilon_{m}).$ They could then sample $\pi'(\mathbf{\beta}) \in \Gamma_{\epsilon}$ using Algorithm A.I. to derive the posterior $\pi'(\mathbf{\beta}|y^{0}_{1},\ldots, y^{0}_{n})$ for each $\pi'$ and thus examine the robustness of their posterior $\pi(\mathbf{\beta}|y^{0}_{1},\ldots, y^{0}_{n})$ to uncertainty in the data.

Thanks to duality, there is yet another way in which the researcher could use the ABC class of priors. Suppose now that the researcher has no concern about the observed data, however, they were not entirely sure that the prior distributions on $\mathbf{\beta}$ were quite accurate. They believe that the prior distributions $\pi(\mathbf{\beta})$ are, say, at most $10\%$ off the mark. They know that this equates (\cite{joshi18}) to $K(\pi,\pi') \leq 0.1.$ Then using \eqref{Kol1}, they determine that they have
\begin{equation} \label{Kol_ex}
 \sup_{\tau\in \Theta} \left|\int^{\tau} \pi(\mathbf{\beta})\left ( 1 -\frac{\exp[ \sum_{i} \mathbf{\beta}'\mathbf{X}_{i} \times t_{i}]}{E_{\pi}\left[ \exp[ \sum_{i} \mathbf{\beta}'\mathbf{X}_{i}\times t_{i}]\right]}\right )\,d\mathbf{\beta}\right| \leq 0.1.
\end{equation}
They can use \eqref{Kol_ex} to (numerically) determine $\pm\mathbf{\epsilon}$ that satisfies the above inequality and thus determine $\Gamma_{\epsilon}.$ They could then sample $\pi'(\mathbf{\beta}) \in \Gamma_{\epsilon}$  using Algorithm A.I. to derive the posterior $\pi'(\mathbf{\beta}|y^{0}_{1},\ldots, y^{0}_{n})$ for each $\pi'$ and thus examine the robustness of their posterior $\pi(\mathbf{\beta}|y^{0}_{1},\ldots, y^{0}_{n})$ to uncertainty in the prior distributions.

While there is likely no unanimity on the definition of the Objective Bayes (OB) analysis (\cite{berger06}), one commonly held view (\cite{bayari07}, \cite{consonni18}) is that an OB method should only use the information contained in the statistical model and no other external information. The ABC class of priors can therefore be considered to be a class of objective prior distributions since the members of the class are defined using the likelihood (or the function of the sufficient statistics) evaluated at the observed data $\mathbf{x^{0}}$ and the maximum deviation from the observed data or sufficient statistics that is permitted. However, note that this class is centered around the original prior distribution $\pi$ and we haven't made any assumption regarding its choice. Thus, the ABC class of priors can be considered to be a class that objectively captures the uncertainty around a possibly subjective prior distribution $\pi.$ As discussed above, this could either be the uncertainty in the observed data or the uncertainty in the elicited prior distributions.

The catch in the definitions of these classes of priors is that the definitions include the likelihood (or the function of the sufficient statistics) evaluated at the observed data $\mathbf{x^{0}}.$ This may seem odd that defining a prior distribution includes the observed data. However, note firstly that, the classes are centered around the original prior $\pi$ which was defined \emph{apriori} and secondly, we are using the ABC classes to perform prior robustness analysis, which is  perfectly reasonable to be done post-data. 

There could also be situations where one wishes to define these classes apriori. In such cases, $\mathbf{x^{0}}$ can be considered to be the data that one \emph{anticipates} observing. One elicits $\mathbf{x^{0}}$ and $\epsilon$ to define any of these classes before observing the data.


For ABC methods, there is, as yet, no objective criteria to decide how big/small $\epsilon$ should be. By showing how, in fact, $\epsilon$ induces the corresponding class of priors may help provide insight into the choice of $\epsilon$.  \cite{joshi18} show how the Kolmogorov distance can be used to elicit possible distortions in the prior distribution.   The Kolmogorov distance function computed in Section \ref{abc-G} can be used to elicit the value of $\epsilon$ that may be acceptable. Thus, this may improve the practical and intuitive understanding of ABC.


The aim of this paper is not computational. As illustrated in Section \ref{algo}, it is possible to sample from the posterior distributions $\pi(\mathbf{\theta}|\mathbf{x'})$  that correspond to the prior distributions $\pi'(\mathbf{\theta})$  that belong to the ABC classes of priors to generate the true posterior distribution. The algorithms described here use importance sampling and as a result will likely only work well for low dimensional $\theta.$ They have been provided here to illustrate the ideas. However, the possible computational advantage of these ideas is that it may be possible to obtain the posterior distribution without having to simulate $\mathbf{x'}$ for each sampled value of $\mathbf{\theta}.$ The algorithms provided here only use the observed data. For problems where data simulation is computationally expensive, this provides a possible alternative for efficient posterior estimation. Of course, to achieve a computational benefit, more efficient algorithms will need to be used to sample from the posterior distributions $\pi(\mathbf{\theta}|\mathbf{x'})$. This could be a potential area for future research.

\section*{Acknowledgements}
Chaitanya Joshi's work was supported by the research funds provided by The University of Waikato. Research was partially performed during the visit of both authors at SAMSI (The Statistical and Applied Mathematical Sciences Institute).

\bibliography{ABC_prior_robustness_Arxiv_April2020}
\bibliographystyle{apalike}

\end{document}